\documentclass[11pt,a4paper]{article}

\usepackage[utf8]{inputenc}
\usepackage[english]{babel}

\usepackage{geometry}
\usepackage{graphicx}
\usepackage{amssymb,amsmath,amsthm}
\usepackage{marvosym}
\usepackage{subcaption}
\usepackage{delarray}
\usepackage{enumitem}
\usepackage{tikz}
\usepackage{hyperref}
\usepackage{ifthen} 
\usepackage{authblk}

\newtheorem{theorem}{\textbf{Theorem}}

\newtheorem{proposition}{\textbf{Proposition}}

\newtheorem{remark}{\textbf{Remark}}

\newtheorem{remark*}{Remark}

\newtheorem{open}{\textbf{Open question}}

\newcommand{\N}{\mathbb{N}}
\newcommand{\Z}{\mathbb{Z}}

\newcommand{\Bool}{\{0,1\}}

\newcommand{\Poly}{{\mathsf{P}}}
\newcommand{\NC}{{\mathsf{NC}}}
\newcommand{\AC}{{\mathsf{AC}}}
\newcommand{\NL}{{\mathsf{NL}}}
\newcommand{\decisionpb}[3]{\fbox{\parbox{0.9\textwidth}{{\bf #1}\\{\it Input:} #2\\{\it Question:} #3}}}
\newcommand{\neighbors}[1]{{\mathcal{N}{({#1})}}}
\newcommand{\indic}{{\bf 1}}

\newcounter{gridcounterx}
\newcounter{gridcountery}

\title{Freezing sandpiles and Boolean threshold networks: equivalence and complexity}
\author[1]{Eric Goles}
\author[1]{Pedro Montealegre}
\author[2,3]{K\'evin Perrot}
\affil[1]{Facultad de Ingenieria y Ciencias, Univ. Adolfo Iba\~{n}ez, Santiago, Chile}
\affil[2]{Aix Marseille Univ., Univ. de Toulon, CNRS, LIS, UMR 7020, Marseille, France}
\affil[3]{Univ. Côte d'Azur, CNRS, I3S, UMR 7271, Sophia Antipolis, France}
\date{}

\begin{document}
\maketitle

\begin{abstract}
  The $\NC$ versus $\Poly$-hard classification of the prediction problem for
  sandpiles on the two dimensional grid with von Neumann neighborhood is a
  famous open problem. In this paper we make two kinds of progresses,
  by studying its freezing variant. First, it enables to establish strong
  connections with other well known prediction problems on networks of threshold
  Boolean functions such as majority. Second, we can highlight some necessary
  and sufficient elements to the dynamical complexity of sandpiles,
  with a surprisingly crucial role of cells with two grains.
\end{abstract}

\section{Introduction}


The sand pile model, as well as the Boolean threshold automata, have been studied and applied extensively in various domains ~\cite{fp19,m97,gm90,gmt13,gm16,gmmo17,gmpt17}.  The classical sandpile model on the two dimensional grid with von Neumann
neighborhood was introduced in the 1980 by Bak, Tang and Wiesenfeld, as a
simple and natural model of some physical phenomena~\cite{btw87}. In~\cite{gm97} Goles and Margenstern showed that in arbitrary graphs, any given Turing machine can be simulated by a configuration of the sandpile model. This means that on an arbitrary topology the dynamic of the sandpile model is Turing-universal. After that, in~\cite{mn99}, Moore and Nilsson started the study of how difficult
it is to predict the behavior of sandpiles, bringing the question to the formal
theory of computational complexity. Sandpile prediction problems are usually
solvable in polynomial time by simply running the simulation until the dynamics
reaches a stable state. Essentially, the results of Moore and Nilsson say that
sandpiles in one dimension are efficiently predictable in parallel (in $\NC$),
and that sandpiles in three dimensions or more are intrinsically sequential
($\Poly$-complete).  It leaves open the two-dimensional case, which has not yet
been closed despite considerable
efforts~\cite{glmmp04,gg06,gmpt17,np18,bmot18}.

Following a trend of research leading to new discoveries around well known open
problems on majority dynamical
systems (reviewed in Section~\ref{s:known}), we introduce in this paper
the {\em freezing} variant of sandpiles, where each site can be fired at most once. Very interestingly, the freezing world
breaks a fundamental barrier between majority and sandpiles. Though it is known
that two-dimensional majority can simulate two-dimensional
sandpiles~\cite{gmpt17}, it is unknown whether the
converse is true. Indeed, the main difficulty lies in the so called {\em
abelian property} of sandpile models (the fact that sand grains may be toppled
in any order), which is absent in the majority rule. This later model
therefore heavily depends on the parallel schedule of cells update
(see~\cite{gm14}), which is not the case in sandpiles and makes a simulation
result hard to establish.  It turns out that the freezing world breaks this
frontier, as majority and other threshold Boolean functions are not
sensitive to the order of cells update in this case. Predicting freezing
variants of dynamical systems may be thought as a ``simplest case'' study of
their complexity (see~\cite[Proposition~6]{fp19} and
Remark~\ref{remark:spp} for a formal discussion).

The complexity classes at stake are $\AC^0$ (constant time
in parallel), $\NL$ (non-deterministic logarithmic space), $\NC$ (poly-logarithmic
time in parallel), and $\Poly$ (polynomial time), with
$\AC^0 \subseteq \NL \subseteq \NC \subseteq \Poly$
(see for example~\cite{ghr95}).

In Section~\ref{s:def} we define the model and problem under consideration and
in Section~\ref{s:known} we review results on the computational complexity of
prediction problems in the freezing world. Section~\ref{s:bool} establishes in
this setting an isomorphism between the dynamics of sandpiles and threshold Boolean
functions on a grid layout. Finally, Section~\ref{s:complexity} studies all
possible restrictions for the freezing sandpile prediction problem, consisting
in allowing only a subset of sand contents in the configuration given as input.
All but two cases are classified as being in $\NC$ or as hard as the general
(freezing) case, which allows to define even simpler sandpile prediction
problems yet preserving the complexity of the general case. The two remaining
cases are discussed at the end of the Section; the difficulty to relate them
with other models brings novel insights on a possible hierarchy of sandpile
prediction problems, between $\NC$ and $\Poly$.

\section{Definitions}
\label{s:def}

We consider the freezing variant of the classical sandpile model introduced by
Bak, Tang and Wiesenfled in~\cite{btw87} on the two dimensional grid with von
Neumann neighborhood. A {\em configuration} $c \in (\N \cup
\{-\infty\})^{\Z^2}$ assigns a number of sand grains to each cell of the grid,
or $-\infty$ when a cell has already fired. For commodity let $c_v$ denote the
sand content at position $v \in \Z^2$ in configuration $c$. When the sand
content of a cell exceeds its number of out-neighbors (four in the grid with
von Neumann neighborhood), then the cell gives one grain to each of its
out-neighbors and enters the state $-\infty$ ({\em freezing state}) so that it
never gives grains again. Formally, with
$\neighbors{(i,j)}=\{(i,j+1),(i+1,j),(i,j-1),(i-1,j)\}$ the dynamics is defined
by $F : (\N \cup \{-\infty\})^{\Z^2} \to (\N \cup \{-\infty\})^{\Z^2}$ such
that for all $v \in \Z^2$,
$$
  F(c)_v=
  \left\{\begin{array}{ll}
    -\infty &\text{ if } c_v \geq 4\\[.5em]
    c_v + \!\!\!\!\sum\limits_{u \in \neighbors{v}} \indic_{\N}(c_u-4) &\text{ otherwise}
  \end{array}\right.
$$
where $\indic_{\N}(x)$ is the indicator function of $\N$, which equals $1$ when
$x \geq 0$, and 0 when $x < 0$, for any $x \in \Z$. Remark that this discrete
dynamical system is deterministic. When a cell gives grains to its neighbors we
say that it {\em fires}, and immediately {\em freezes}.
 
Classically, a configuration $c$ is {\em finite} when the number of non-empty
cells is finite, that is when $|\{v \mid c_v \neq 0\}| < \infty$. Note that up
to translation, all non-empty cells of a finite configuration can always be
placed inside a rectangle of (finite) size $n \times m$ with the bottom left
corner at the origin, hence going from $(0,0)$ to $(n-1,m-1)$. Such rectangular
non-empty parts of finite configurations will be given as inputs to the problem
we consider. For the purpose of this article, since we will restrict the
allowed values on configurations and sometimes forbid the value $0$, we define
a {\em finite} configuration $c$ as having the freezing state outside the
encompassing rectangle, that is with $c_{(v_x,v_y)}=-\infty$ when $v_x < 0$ or
$v_x \geq n$ or $v_y < 0$ or $v_y \geq m$. A configuration $c$ is {\em stable}
when no grain moves, that is when $c_v < 4$ for all cells $v$. Additionally, we
say that a finite configuration $c$ is {\em simple} when for all cell $v$
inside the rectangle of size $n \times m$ we have $c_v \in \{0,1,2,3,4\}$.\\

\decisionpb{Freezing sandpiles prediction problem (FSPP)}{a simple finite
configuration $c$ and a cell $v$.}{does there exist $t$ such that $F^t(c)_v
\geq 4?$}\\

It is straightforward to notice that the problem is solvable in quadratic time
({\bf FSPP} $\in \Poly$) by running the simulation: one step takes a linear
time to be computed, and at least one cell freezes at each step or we have
reached a stable configuration.  Since cells remain frozen, after linearly many
steps the configuration is stable, and we can answer.

Let {\bf Sandpiles prediction problem} ({\bf SPP}) be the analogous prediction
problem on classical (non-freezing) sandpile model, the one not known to be in
$\NC$ nor $\Poly$-hard. We denote $\leq^{m}_{\NC}$ the many-one reduction in
$\NC$, and $\leq^{m}_{\AC^0}$ the many-one reduction in $\AC^0$.

\begin{remark}
  \label{remark:spp}
  As stated in~\cite{np18} (Lemma 1), when there is only one value $4$, and if
  furthermore this value $4$ is placed on the border of the $n \times m$
  rectangle containing the finite configuration, then {\em SPP} is the same as
  {\em FSPP} because each cell is fired at most once. 
\end{remark}

\section{Known results}
\label{s:known}

For the sandpile model~\cite{mn99} and majority cellular automata~\cite{m97},
it has early been proven that the prediction problem is in $\NC$ for
dimension one, and $\Poly$-hard for dimension three and above. Such
$\Poly$-hardness results~\cite{m93,mn97,nw06} are commonly proven via
reductions from the canonical {\em circuit value problem} (CVP) originally
proven to be $\Poly$-complete by Ladner~\cite{l75}, or its monotone variant
(MCVP), or its planar variant (PCVP) (see~\cite{ghr95}). It is remarkable that
all these reductions employ {\em Bank's encoding}~\cite{b71}: dynamical chains
of reactions implement circuit computations on a quiescent background, as
electrons moving along wires. Note that planar monotone circuit value problem
(PMCVP), which is easily reducible to sandpiles and majority prediction
problems (even freezing), has however been proven to be in $\NC$, {\em i.e.}
efficiently computable in parallel~\cite{y91}.

Studies of the freezing world have been introduced in cellular
automata~\cite{got15}, where the authors prove that Turing universality can be
achieved even in one dimension, and that the prediction problem may be
$\Poly$-complete in two dimensions, though in one dimension it is in $\NL$. On
the two-dimensional grid with von Neumann neighborhood, it is proven
in~\cite{bmot18} that at least two state changes are necessary to be
intrinsically universal according to block simulation, and furthermore that two
state changes are sufficient. The paper~\cite{ot19} places freezing cellular
automata universality and prediction in the broader context of bounded-change
and convergent cellular automata.

Previous works on the prediction of threshold (majority-like) functions will be
useful in our analysis of sandpile prediction problems complexity.
It is known that predicting {\em freezing strict majority} is $\Poly$-complete
for undirected graphs of maximum degree at least five ($\Delta(G) \geq 5$), and
in $\NC$ for undirected graphs of maximum degree at most four ($\Delta(G) \leq
4$)~\cite{gmt13}.
Regarding {\em freezing non-strict majority}, its prediction is
$\Poly$-complete for undirected graphs of maximum degree at least four
($\Delta(G) \geq 4$), and in $\NC$ for undirected graphs of maximum degree at
most three ($\Delta(G) \leq 3$)~\cite{gmt13}.
This latter has also been proven to be in $\NC$ for the two-dimensional grid
with von Neumann neighborhood~\cite{gmmo17}.
Remark that, although planarity is known to forbid information crossing on
sandpile models~\cite{gg06,np18}, it is an obstacle that can be overcome on
non-freezing majority~\cite{gm15} (based on a planar traffic light gadget
of degree five, exploiting the non-freeziness).

\section{Sandpiles as a patchwork of threshold Boolean functions}
\label{s:bool}

We begin with a remark linking the dynamics of (freezing) sandpiles to that of
an assembly of Boolean functions. These relations will be useful to
employ the literature in order to prove that some problems are in $\NC$
(Section~\ref{s:complexity}). Indeed, the dynamics of finite freezing sandpiles
can be seen as a grid network of freezing threshold Boolean functions.

Let us define finite freezing Boolean networks on the two dimensional
grid with von Neumann neighborhood. Let $G_{n \times m}=(V_{n \times m},E_{n
\times m})$ be the finite undirected graph defined as the subgraph of the two
dimensional grid induced by vertices in the rectangle of size $n \times m$ with
the bottom left corner at the origin. Formally, 
\begin{align*}
  V_{n \times m} &= \{ (x,y) \in \Z^2 \mid 0 \leq x \leq n-1 \text{ and } 0 \leq y \leq n-1 \}\\
  E_{n \times m} &= \{ (u,v) \in V^2 \mid v \in \neighbors{u} \}
\end{align*}
For simplicity, when the dimensions are clear from the context, we will denote
$G=(V,E)$ such a graph. The set of configurations is $\Bool^V$, and each vertex
$v$ is equipped with a local Boolean function which is freezing (state 1 is
always sent to state 1). We use five such local functions, given a
configuration $c$:
\begin{itemize}
  \item $\wedge$ ({\em and}) defined as
    $f^{\wedge}_v(c)=\left\{\begin{array}{ll}
      1 & \text{ if } c_v=1 \text{ or } \sum\limits_{u \in \neighbors{v}} c_u = 4\\
      0 & \text{ otherwise}
    \end{array}\right.$
  \item $M$ ({\em strict majority}) defined as
    $f^{M}_v(c)=\left\{\begin{array}{ll}
      1 & \text{ if } c_v=1 \text{ or } \sum\limits_{u \in \neighbors{v}} c_u > 2\\
      0 & \text{ otherwise}
    \end{array}\right.$
  \item $m$ ({\em non-strict majority}) defined as
    $f^{m}_v(c)=\left\{\begin{array}{ll}
      1 & \text{ if } c_v=1 \text{ or } \sum\limits_{u \in \neighbors{v}} c_u \geq 2\\
      0 & \text{ otherwise}
    \end{array}\right.$
  \item $\vee$ ({\em or}) defined as
    $f^{\vee}_v(c)=\left\{\begin{array}{ll}
      1 & \text{ if } c_v=1 \text{ or } \sum\limits_{u \in \neighbors{v}} c_u \geq 1\\
      0 & \text{ otherwise}
    \end{array}\right.$
  \item $1$ ({\em constant $1$}) defined as $f^{1}_v(c)=1$
\end{itemize}

Note that each local function only depends on the state of the vertex and its
neighbors, and is invariant by permutation of the neighbors, as is the case in
freezing sandpiles. Also, the formulation of local functions takes into account
the fact that some vertices on the border of the graph $G$ are missing some
neighbors (the number of neighbors, four, is hard-coded in the local
functions). Let $B : \Bool^V \to \Bool^V$ be the dynamics obtained by applying
in parallel the local function assigned to each vertex, {\em i.e.} such that
for all $v \in V$ we have
$$
  B(c)_v = f_v(c).
$$

Given a finite simple sandpile configuration $c$ of size $n \times m$ for the
freezing sandpile model, we define the corresponding freezing threshold Boolean
network $B_c$ of size $n \times m$ (where the local function $f_v$  at $v$
depends on the value of $c_v$) by
\begin{center}
  \begin{tabular}{c|ccccc}
    $c_v$ & 0 & 1 & 2 & 3 & 4\\
    \hline
    $f_v$ & $\wedge$ & $M$ & $m$ & $\vee$ & $1$
  \end{tabular}\\[.2em]
\end{center}
and a configuration on this network as
$$
  \phi(c)_v=\left\{\begin{array}{ll} 1 \text{ if } c_v=-\infty\\ 0 \text{ ortherwise.}\end{array}\right.
$$

Given any finite simple sandpile configuration $c$ of size $n \times m$, we
obtain the freezing Boolean network $B_c$ of size $n \times m$, which dynamics
commutes with the transformation $\phi$ on configurations.

\begin{proposition}\label{prop:bool}
  For any finite simple $c$ and all $t \in \N$ we have $B^t_c(\phi(c)) = \phi(F^t(c))$.
\end{proposition}

\begin{proof}
  Note that initially freezed sandpile cells (outside the rectangle of size $n
  \times m$) are discarded in the corresponding freezing Boolean network.
  Starting from the initial configuration with all cells in Boolean state $0$,
  this latter will therefore simply transform a grain move from $u$ to $v$ into
  the fact that $u$ is a neighbor of $v$ in state $1$. At each time step, a
  cell in the freezing state remains in the freezing state in both dynamics.
  Regarding other cells, in both dynamics and at each step, they enter the
  freezing state if and only if at least the same number (in both dynamics) of
  their neighbors are in the freezing state. And $\phi$ depends only on the
  freezing (or not) state of each cell.
\end{proof}

\section{Computational complexity of {\bf FSPP}}
\label{s:complexity}

We study the computational complexity of restrictions on {\bf FSPP}, depending
on the sand contents that each cell of the simple finite configuration given as
input can take among $\{0,1,2,3,4\}$. It is obvious that forbidding the value
$4$ leads to answering {\em no} to any prediction question, and allowing only
the values $3$ and $4$ (or just $4$) leads to always answering {\em yes},
therefore we consider only the 14 remaining cases. For any $A \subseteq
\{0,1,2,3,4\}$ we say that a configuration $c$ is {\em $A$-simple} when for all
cell $v$ we have $c_v \in A$. With this notation, simple means
$\{0,1,2,3,4\}$-simple.\\

\decisionpb{$A$-freezing sandpiles prediction problem ($A$-FSPP)}{an $A$-simple
finite configuration $c$ and a cell $v$.}{does there exist $t$ such that
$F^t(c)_v \geq 4?$}\\

Let us underline that the restriction to $A \subseteq \{0,1,2,3,4\}$ is sound,
as follows.

\begin{proposition}
  If we generalize the definition of $A$-simple configuration, then we still
  have $A${\bf -FSPP} $\leq^m_{\AC^0}$ {\bf FSPP} for any finite $A \subseteq
  \N$.
\end{proposition}

\begin{proof}
  Let $(c,v)$ be an instance of $A${\bf -FSPP}. Since the model is freezing,
  the cell to cell transformation of $c$ into $c'$ defined as $c_v \mapsto
  \min\{c_v,4\}$ (note that the outer part of the finite rectangle is not
  modified)
  preserves the answer (after one step we have $F(c)=F(c')$), and $(c',v)$ is
  an instance of {\bf FSPP}.
\end{proof}

Considering $A$ among
\begin{center}
  $\{0,4\},\{1,4\},\{0,1,4\},\{2,4\},\{0,2,4\},\{1,2,4\},\{0,1,2,4\},$
  $\{0,3,4\},\{1,3,4\},\{0,1,3,4\},\{2,3,4\},\{0,2,3,4\},\{1,2,3,4\},\{0,1,2,3,4\},$
\end{center}
the results are summed-up in Theorems~\ref{theorem:nc} and~\ref{theorem:fspp},
plus the Open question~\ref{question:0134}.

\begin{theorem}\label{theorem:nc}
  {\bf $A$-FSPP} $\in \NC$ when $A$ is one of
  $$\{0,4\},\{1,4\},\{0,1,4\},\{2,4\},\{0,3,4\},\{2,3,4\}.$$
\end{theorem}

\begin{theorem}\label{theorem:fspp}
  {\bf FSPP} $\leq^{m}_{\AC^0}$ {\bf $A$-FSPP} when $A$ is one of
  $$\{0,2,4\},\{1,2,4\},\{0,1,2,4\},\{0,2,3,4\},\{1,2,3,4\}\text{ and }\{0,1,2,3,4\}.$$
\end{theorem}

\begin{open}
  \label{question:0134}
  {\bf $\{0,1,3,4\}$-FSPP} $\leq^{m}_{\AC^0}$ {\bf $\{1,3,4\}$-FSPP},
  but does {\bf $\{1,3,4\}$-FSPP} $\in \NC$
  or {\bf FSPP} $\leq^{m}_{\AC^0}$ {\bf $\{0,1,3,4\}$-FSPP}?
\end{open}

Subsections~\ref{ss:nc} and~\ref{ss:fspp} will present respectively the results
of Theorems~\ref{theorem:nc} and~\ref{theorem:fspp}.  Subsection~\ref{ss:0134}
will present some perspectives on Open question~\ref{question:0134}.

\subsection{Restrictions efficiently predictable in parallel}
\label{ss:nc}

This section makes use of the developments presented in Section~\ref{s:bool},
in order to apply results from the literature on problems in $\NC$.

\subsubsection{$\{0,4\}${\bf -FSPP}}

When {\bf FSPP} is restricted to $\{0,4\}$-simple configurations, according to
Proposition~\ref{prop:bool} it corresponds to a finite freezing Boolean network
on the grid with only {\em and} and {\em constant $1$} local functions, which
can be decided in constant parallel time: the instance $(c,v)$ is positive if
and only if $\phi(c)_v=1$ or $\sum_{u \in \neighbors{v}} \phi(c)_u = 4$ (with
$|\neighbors{v}| \leq 4$).

\begin{proposition}\label{prop:04}
  $\{0,4\}${\bf -FSPP} $\in \AC^0$.
\end{proposition}

\subsubsection{$\{1,4\}${\bf -FSPP} and $\{0,1,4\}${\bf -FSPP}}
\label{ss:014}

When {\bf FSPP} is restricted to $\{1,4\}$-simple configurations, according to
Proposition~\ref{prop:bool} it corresponds to a finite freezing Boolean network
on the grid with only {\em strict majority} and {\em constant $1$} local
functions. However, {\em constant $1$} local functions are the same as {\em
strict majority} cells initially in state $1$ since we are in a freezing world.
As a consequence, we are left with only {\em strict majority} local functions
on a grid with von Neumann neighborhood, which can be predicted in $\NC^2$
according to~\cite{gmt13} (to adapt the setting it is sufficient to add a
border of cells in state $0$). The transformation is easily performed in
$\AC^0$, leading to an overall algorithm in $\NC^2$. See Figure~\ref{fig:14}
for an illustration.

\begin{figure}
  \centerline{\begin{tikzpicture}[scale=.5]
  \draw[black!30] (0,0) grid (7,5);
  \foreach \x in {0,...,6}{
    \node[black!30] at (\x+.5,.5) {\tiny $-\infty$};
    \node[black!30] at (\x+.5,4.5) {\tiny $-\infty$};
  }
  \foreach \y in {1,2,3}{
    \node[black!30] at (.5,\y+.5) {\tiny $-\infty$};
    \node[black!30] at (6.5,\y+.5) {\tiny $-\infty$};
  }
  \draw (1,1) grid (6,4);
  \setcounter{gridcounterx}{1} 
  \setcounter{gridcountery}{1} 
  \foreach \v in {1,4,4,4,1,4,1,1,1,1,4,4,4,4,4}{
    \node at (\value{gridcounterx}+.5,\value{gridcountery}+.5) {$\v$};
    \addtocounter{gridcounterx}{1}
    \ifthenelse{\value{gridcounterx}>5} 
    {
      \setcounter{gridcounterx}{1} 
      \addtocounter{gridcountery}{1}
    }{}
  }
  \node at (8,2.5) {$\mapsto$};
  \draw (9,0) grid (16,5);
  \setcounter{gridcounterx}{9} 
  \setcounter{gridcountery}{0} 
  \foreach \v in {0,0,0,0,0,0,0,0,0,1,1,1,0,0,0,1,0,0,0,0,0,0,1,1,1,1,1,0,0,0,0,0,0,0,0}{
    \node at (\value{gridcounterx}+.5,\value{gridcountery}+.5) {$\v$};
    \addtocounter{gridcounterx}{1}
    \ifthenelse{\value{gridcounterx}>15} 
    {
      \setcounter{gridcounterx}{9} 
      \addtocounter{gridcountery}{1}
    }{}
  }
\end{tikzpicture}}
  \caption{Transformation in $\AC^0$ of a $\{1,4\}$-simple sandpile
  configuration to a configuration for the {\em freezing strict majority}
  dynamics on the grid~\cite{gmt13}.}
  \label{fig:14}
\end{figure}
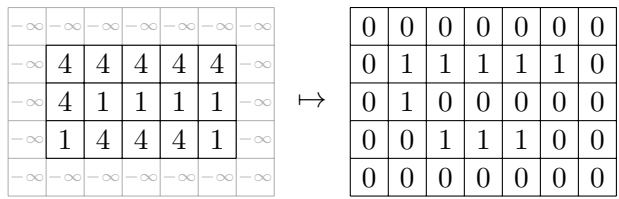

\begin{proposition}\label{prop:14}
  $\{1,4\}${\bf -FSPP} $\in \NC^2$.
\end{proposition}

The result of~\cite{gmt13} can also be applied to prove that $\{0,1,4\}${\bf
-FSPP} is in $\NC$. The idea is that cells $u$ with $c_u=0$ are completely
passive in the freezing dynamics (they fire if and only if all their four
neighbors are already fired and frozen). More precisely, given an instance
$(c,v)$ we consider two cases.
\begin{enumerate}
  \item If $c_v \neq 0$ then we perform in $\AC^0$ the following modification
    of the grid: each vertex $u=(u_x,u_y)$ such that $c_u=0$ is replaced with
    four vertices $u_n,u_e,u_s,u_w$ and the arcs $\{(u_x,u_y+1),u_n\}$,
    $\{(u_x+1,u_y),u_e\}$, $\{(u_x,u_y-1),u_s\}$, $\{(u_x-1,u_y),u_w\}$, and
    $\{u_n,u_e\}$, $\{u_e,u_s\}$, $\{u_s,u_w\}$, $\{u_w,u_n\}$. With state $1$
    on vertices $u$ such that $c_u=4$ and state $0$ elsewhere, answers to the
    prediction under {\em freezing strict majority} on this graph $G$ and to
    the freezing sandpiles prediction problem are identical. Indeed, in the
    {\em freezing strict majority} dynamics the newly created vertices
    corresponding to cells such that $c_u=0$ will never reach state $1$ because
    they always have two of their three neighbors in state $0$. Since cells
    such that $c_u=0$ are completely passive in the sandpile dynamics ({\em
    i.e.} considering that they do not fire leaves the behavior of other cells
    unchanged), and since the questioned cell $v$ is not one of these, $v$ will
    fire from $c$ if and only if it reaches state $1$ in the {\em freezing
    strict majority} dynamics on $G$. Finally, we have $\Delta(G) \leq 4$,
    therefore~\cite{gmt13} gives an $\NC^2$ algorithm to predict the {\em
    freezing strict majority} dynamics. See an example on Figure~\ref{fig:014}.
  \item If $c_v = 0$ then, if furthermore at least one of the four neighbors of
    $v$ is $0$ (or $-\infty$) then $v$ cannot fire and the answer is negative.
    Otherwise we do the same transformation as in the case $c_v \neq 0$, and
    ask if each of the four neighbors of $v$ will fire (still in $\NC^2$). The
    answer for $v$ is positive (it will fire) if and only if all its four
    neighbors will fire.
\end{enumerate}

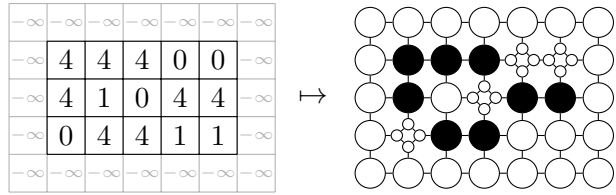
\begin{figure}
  \centerline{\begin{tikzpicture}[scale=.5]
  \draw[black!30] (0,0) grid (7,5);
  \foreach \x in {0,...,6}{
    \node[black!30] at (\x+.5,.5) {\tiny $-\infty$};
    \node[black!30] at (\x+.5,4.5) {\tiny $-\infty$};
  }
  \foreach \y in {1,2,3}{
    \node[black!30] at (.5,\y+.5) {\tiny $-\infty$};
    \node[black!30] at (6.5,\y+.5) {\tiny $-\infty$};
  }
  \draw (1,1) grid (6,4);
  \setcounter{gridcounterx}{1} 
  \setcounter{gridcountery}{1} 
  \foreach \v in {0,4,4,1,1,4,1,0,4,4,4,4,4,0,0}{
    \node at (\value{gridcounterx}+.5,\value{gridcountery}+.5) {$\v$};
    \addtocounter{gridcounterx}{1}
    \ifthenelse{\value{gridcounterx}>5} 
    {
      \setcounter{gridcounterx}{1} 
      \addtocounter{gridcountery}{1}
    }{}
  }
  \node at (8,2.5) {$\mapsto$};
  \tikzstyle{maj} = [draw,circle,inner sep=4pt]
  \tikzstyle{majsmall} = [draw,circle,inner sep=1.5pt]
  \foreach \x in {0,...,6}{
    \node[maj] (maj-\x-0) at (9.5+\x,.5) {};
    \node[maj] (maj-\x-4) at (9.5+\x,4.5) {};
  }
  \foreach \y in {1,2,3}{
    \node[maj] (maj-0-\y) at (9.5,.5+\y) {};
    \node[maj] (maj-6-\y) at (15.5,.5+\y) {};
  }
  \foreach \x/\y in {1/1,3/2,4/3,5/3}{
    \foreach \xx/\yy in {0/1,1/0,0/-1,-1/0}{
      \node[majsmall] (maj-\x-\y-\xx-\yy) at (9.5+\x+\xx*.3,.5+\y+\yy*.3) {};
    }
  }
  \foreach \x/\y in {4/1,5/1,2/2}{
    \node[maj] (maj-\x-\y) at (9.5+\x,.5+\y) {};
  }
  \foreach \x/\y in {2/1,3/1,1/2,4/2,5/2,1/3,2/3,3/3}{
    \node[maj,fill=black] (maj-\x-\y) at (9.5+\x,.5+\y) {};
  }
  \foreach \x/\xx/\y in {
    0/1/0,1/2/0,2/3/0,3/4/0,4/5/0,5/6/0,%
    2/3/1,3/4/1,4/5/1,5/6/1,%
    0/1/2,1/2/2,4/5/2,5/6/2,%
    0/1/3,1/2/3,2/3/3,%
    0/1/4,1/2/4,2/3/4,3/4/4,4/5/4,5/6/4%
  }{
    \draw (maj-\x-\y) to (maj-\xx-\y);
  }
  \foreach \x/\y/\yy in {
    0/0/1,0/1/2,0/2/3,0/3/4,%
    1/2/3,1/3/4,%
    2/0/1,2/1/2,2/2/3,2/3/4,%
    3/0/1,3/3/4,%
    4/0/1,4/1/2,%
    5/0/1,5/1/2,%
    6/0/1,6/1/2,6/2/3,6/3/4%
  }{
    \draw (maj-\x-\y) to (maj-\x-\yy);
  }
  \foreach \x/\y in {1/1,3/2,4/3,5/3}{
    \draw (maj-\x-\y-0-1) to (maj-\x-\y-1-0);
    \draw (maj-\x-\y-1-0) to (maj-\x-\y-0--1);
    \draw (maj-\x-\y-0--1) to (maj-\x-\y--1-0);
    \draw (maj-\x-\y--1-0) to (maj-\x-\y-0-1);
  }
  \draw (maj-1-1-0-1) to (maj-1-2);
  \draw (maj-1-1-1-0) to (maj-2-1);
  \draw (maj-1-1-0--1) to (maj-1-0);
  \draw (maj-1-1--1-0) to (maj-0-1);
  \draw (maj-3-2-0-1) to (maj-3-3);
  \draw (maj-3-2-1-0) to (maj-4-2);
  \draw (maj-3-2-0--1) to (maj-3-1);
  \draw (maj-3-2--1-0) to (maj-2-2);
  \draw (maj-4-3-0-1) to (maj-4-4);
  \draw (maj-4-3-1-0) to (maj-5-3--1-0);
  \draw (maj-4-3-0--1) to (maj-4-2);
  \draw (maj-4-3--1-0) to (maj-3-3);
  \draw (maj-5-3-0-1) to (maj-5-4);
  \draw (maj-5-3-1-0) to (maj-6-3);
  \draw (maj-5-3-0--1) to (maj-5-2);
\end{tikzpicture}}
  \caption{Transformation in $\AC^0$ of a $\{0,1,4\}$-simple sandpile
  configuration to a graph (of maximum degree 4) for the {\em freezing strict
  majority} dynamics~\cite{gmt13}. Vertices in state 0 (resp. 1) are white
  (resp. black).}
  \label{fig:014}
\end{figure}


\begin{proposition}\label{prop:014}
  $\{0,1,4\}${\bf -FSPP} $\in \NC^2$.
\end{proposition}

\subsubsection{$\{2,4\}${\bf -FSPP}}

When {\bf FSPP} is restricted to $\{2,4\}$-simple configurations, according to
Proposition~\ref{prop:bool} it corresponds to a finite freezing Boolean network
on the grid with only {\em non-strict majority} and {\em constant $1$} local
functions, which are the same as {\em non-strict majority} cells initially in
state 1 since we are in a freezing world, and can be decided in $\NC^2$
according to~\cite{gmmo17}.

\begin{proposition}\label{prop:24}
  $\{2,4\}${\bf -FSPP} $\in \NC^2$.
\end{proposition}

\subsubsection{$\{0,3,4\}${\bf -FSPP}}
\label{ss:034}

When {\bf FSPP} is restricted to $\{0,3,4\}$-simple configurations, according
to Proposition~\ref{prop:bool} it corresponds to a finite freezing Boolean
network on the grid with only {\em and}, {\em or} and {\em constant $1$} local
functions. Given an instance $(c,v)$, we consider three cases.
\begin{enumerate}
  \item If $c_v \neq 0$, then we can simply remove the vertices $v$ with
    $c_v=0$ from the graph supporting the finite freezing Boolean network
    dynamics since they are completely passive (they freeze to 1 if and only if
    their four neighbors are already fired and frozen). This construction is
    done in $\AC^0$ and comes down to deciding if there is a path from a cell
    in state $1$ to $v$, which can be done in $\NL$ (choose
    non-deterministically a starting cell in state $1$ and travel
    non-deterministically through a path of length at most $nm$).
  \item If $c_v = 0$ and $c_u \neq 0$ for all $u \in \neighbors{v}$, then we
    compute sequentially the answers of the four instances $(c,u)$ for $u \in
    \neighbors{v}$ (still in $\NL$), and answer positively if and only if all
    these four instances are positive.
  \item If $c_v = 0$ and $c_u = 0$ for at least one $u \in \neighbors{v}$ then
    we can answer negatively: $v$ needs $u$ to go to state $1$ first (strictly
    before $v$ does), and conversely.
\end{enumerate}
Deciding in which of these three cases we are and answering it gives an
algorithm in $\NL$ for $\{0,3,4\}${\bf -FSPP} .

\begin{proposition}\label{prop:034}
  $\{0,3,4\}${\bf -FSPP} $\in \NL$.
\end{proposition}

\subsubsection{$\{2,3,4\}${\bf -FSPP}}

The idea is to reduce the question on a $\{2,3,4\}$-simple configuration to a
question on a $\{2,4\}$-simple configuration, still on the grid. The
transformation is presented on Figure~\ref{fig:234}, each cell at position
$(u_x,u_y) \in \Z^2$ of the $\{2,3,4\}$-simple configuration is transformed
into a macrocell of size $5 \times 6$ whose bottom left corner is at position
$(5u_x,6u_y)$. The questioned cell is placed on the bottom left corner of the
corresponding macrocell (other positions are possible). This reduction can be
computed in constant parallel time, {\em i.e.} in $\AC^0$.

The correctness of the reduction is easily deduced from the abelian property of
sandpiles (the fact that, when the dynamics converges to a stable
configuration, it converges to the same stable configuration regardless of the
order in which firings are performed, in parallel or
sequentially~\cite{glmmp04}). Indeed, if we first consider the firing of values
four in the macrocell corresponding to cell with three grains, then a firing
can occur on the $\{2,3,4\}$-simple configuration if and only if the whole
corresponding macrocell can be fired (otherwise, none of the macrocell's cells
is fired, appart from the initialy fired values four in the macrocells
corresponding to cells with three grains). The result follows by induction.

\begin{figure}
  \centerline{\begin{tikzpicture}[scale=.5]
  \draw (0,0) rectangle node{$2$} ++ (1,1);
  \node at (2,.5) {$\mapsto$};
  \draw (3,-2) grid (9,3);
  \setcounter{gridcounterx}{3} 
  \setcounter{gridcountery}{-2} 
  \foreach \v in {2,2,2,2,2,2,2,2,2,2,2,2,2,2,2,2,2,2,2,2,2,2,2,2,2,2,2,2,2,2}{
    \node at (\value{gridcounterx}+.5,\value{gridcountery}+.5) {$\v$};
    \addtocounter{gridcounterx}{1}
    \ifthenelse{\value{gridcounterx}>8} 
    {
      \setcounter{gridcounterx}{3} 
      \addtocounter{gridcountery}{1}
    }{}
  }
\end{tikzpicture}\hspace*{.4cm}\begin{tikzpicture}[scale=.5]
  \draw (0,0) rectangle node{$3$} ++ (1,1);
  \node at (2,.5) {$\mapsto$};
  \foreach \x/\y in {0/1,0/3,1/0,1/1,1/3,1/4,2/1,2/2,2/3,3/1,3/2,3/3,4/0,4/1,4/3,4/4,5/1,5/3}
    \fill[black!30] (3+\x,-2+\y) rectangle ++ (1,1);
  \draw (3,-2) grid (9,3);
  \setcounter{gridcounterx}{3} 
  \setcounter{gridcountery}{-2} 
  \foreach \v in {2,2,4,4,2,2,2,2,2,2,2,2,4,4,2,2,4,4,2,2,2,2,2,2,2,2,4,4,2,2}{
    \node at (\value{gridcounterx}+.5,\value{gridcountery}+.5) {$\v$};
    \addtocounter{gridcounterx}{1}
    \ifthenelse{\value{gridcounterx}>8} 
    {
      \setcounter{gridcounterx}{3} 
      \addtocounter{gridcountery}{1}
    }{}
  }
\end{tikzpicture}\hspace*{.4cm}\begin{tikzpicture}[scale=.5]
  \draw (0,0) rectangle node{$4$} ++ (1,1);
  \node at (2,.5) {$\mapsto$};
  \draw (3,-2) grid (9,3);
  \setcounter{gridcounterx}{3} 
  \setcounter{gridcountery}{-2} 
  \foreach \v in {4,4,4,4,4,4,4,4,4,4,4,4,4,4,4,4,4,4,4,4,4,4,4,4,4,4,4,4,4,4}{
    \node at (\value{gridcounterx}+.5,\value{gridcountery}+.5) {$\v$};
    \addtocounter{gridcounterx}{1}
    \ifthenelse{\value{gridcounterx}>8} 
    {
      \setcounter{gridcounterx}{3} 
      \addtocounter{gridcountery}{1}
    }{}
  }
\end{tikzpicture}}
  \caption{Cell to macrocell correspondence in the reduction from
  $\{2,3,4\}${\bf -FSPP} to $\{2,4\}${\bf -FSPP}. After one step the grey cells
  in the macrocell corresponding to a cell with three grains have three grains
  (they are neighbor of exactly one cell with four grains).}
  \label{fig:234}
\end{figure}
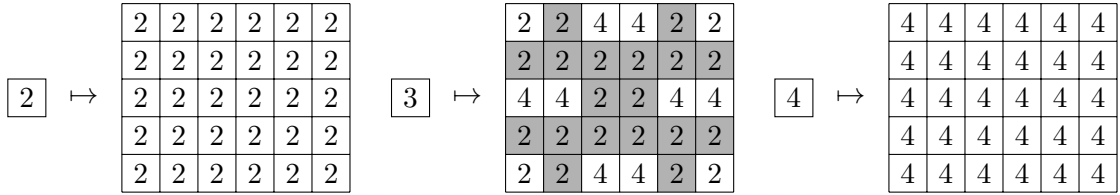

\begin{proposition}\label{prop:234}
  $\{2,3,4\}${\bf -FSPP} $\leq^m_{\AC^0}$ $\{2,4\}${\bf -FSPP}, therefore from
  Propostion~\ref{prop:24} we have $\{2,3,4\}${\bf -FSPP} $\in \NC^2$.
\end{proposition}

\subsection{Restrictions as hard to predict as {\bf FSPP}}
\label{ss:fspp}

We begin with a trivial remark that {\bf FSPP} $\leq^m_{\AC^0}$
$\{0,1,2,3,4\}${\bf -FSPP} with the identity function since the two problems
are identical. We treat subsequent cases one by one, and always use the same
reduction technique: a cell of an input for {\bf FSPP} is converted to a {\em
macrocell} ({\em i.e.} a fixed size rectangle of cells) of an input for $A${\bf
-FSPP}, in constant time and in parallel.

\subsubsection{$\{1,2,3,4\}${\bf -FSPP}}

The reduction is defined as follows: given an instance $(c,v)$ of {\bf FSPP}, we
replace each vertex $(u_x,u_y) \in \Z^2$ of $c$ with a macrocell of size $5
\times 5$ whose bottom left corner is at position $(5u_x,5u_y)$. The cell to
macrocell correspondence is given on Figure~\ref{fig:1234}. This reduction can
be computed in constant parallel time, {\em i.e.} in $\AC^0$. Let us denote
$c'$ the obtained configuration with $v'$ the new questioned cell.

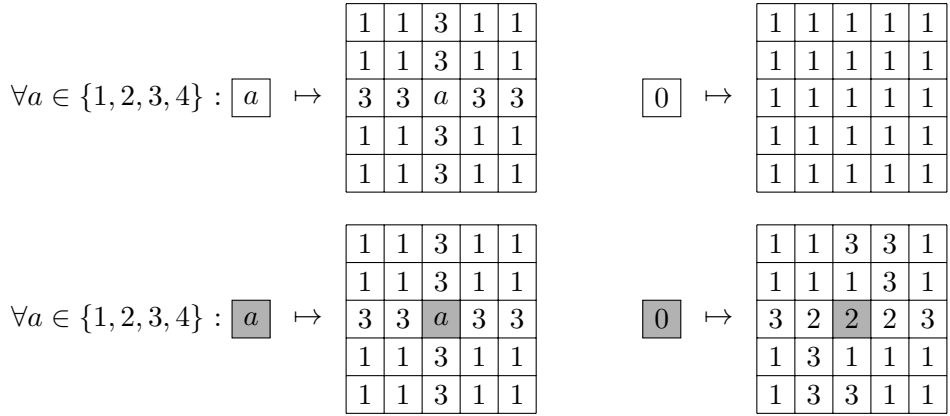
\begin{figure}
  \centerline{
    \begin{tikzpicture}[scale=.5]
  \node[left] at (0,.5) {$\forall a \in \{1,2,3,4\}:$};
  \draw (0,0) rectangle node{$a$} ++ (1,1);
  \node at (2,.5) {$\mapsto$};
  \draw (3,-2) grid (8,3);
  \setcounter{gridcounterx}{3} 
  \setcounter{gridcountery}{-2} 
  \foreach \v in {1,1,3,1,1,1,1,3,1,1,3,3,a,3,3,1,1,3,1,1,1,1,3,1,1}{
    \node at (\value{gridcounterx}+.5,\value{gridcountery}+.5) {$\v$};
    \addtocounter{gridcounterx}{1}
    \ifthenelse{\value{gridcounterx}>7} 
    {
      \setcounter{gridcounterx}{3} 
      \addtocounter{gridcountery}{1}
    }{}
  }
\end{tikzpicture}
    \hspace*{1cm}
    \begin{tikzpicture}[scale=.5]
  \draw (0,0) rectangle node{$0$} ++ (1,1);
  \node at (2,.5) {$\mapsto$};
  \draw (3,-2) grid (8,3);
  \setcounter{gridcounterx}{3} 
  \setcounter{gridcountery}{-2} 
  \foreach \v in {1,1,1,1,1,1,1,1,1,1,1,1,1,1,1,1,1,1,1,1,1,1,1,1,1}{
    \node at (\value{gridcounterx}+.5,\value{gridcountery}+.5) {$\v$};
    \addtocounter{gridcounterx}{1}
    \ifthenelse{\value{gridcounterx}>7} 
    {
      \setcounter{gridcounterx}{3} 
      \addtocounter{gridcountery}{1}
    }{}
  }
\end{tikzpicture}
  }
  \vspace*{1em}
  \centerline{
    \begin{tikzpicture}[scale=.5]
  \node[left] at (0,.5) {$\forall a \in \{1,2,3,4\}:$};
  \fill[black!30] (0,0) rectangle ++ (1,1); 
  \draw (0,0) rectangle node{$a$} ++ (1,1);
  \node at (2,.5) {$\mapsto$};
  \fill[black!30] (5,0) rectangle ++ (1,1); 
  \draw (3,-2) grid (8,3);
  \setcounter{gridcounterx}{3} 
  \setcounter{gridcountery}{-2} 
  \foreach \v in {1,1,3,1,1,1,1,3,1,1,3,3,a,3,3,1,1,3,1,1,1,1,3,1,1}{
    \node at (\value{gridcounterx}+.5,\value{gridcountery}+.5) {$\v$};
    \addtocounter{gridcounterx}{1}
    \ifthenelse{\value{gridcounterx}>7} 
    {
      \setcounter{gridcounterx}{3} 
      \addtocounter{gridcountery}{1}
    }{}
  }
\end{tikzpicture}
    \hspace*{1cm}
    \begin{tikzpicture}[scale=.5]
  \fill[black!30] (0,0) rectangle ++ (1,1); 
  \draw (0,0) rectangle node{$0$} ++ (1,1);
  \node at (2,.5) {$\mapsto$};
  \fill[black!30] (5,0) rectangle ++ (1,1); 
  \draw (3,-2) grid (8,3);
  \setcounter{gridcounterx}{3} 
  \setcounter{gridcountery}{-2} 
  \foreach \v in {1,3,3,1,1,1,3,1,1,1,3,2,2,2,3,1,1,1,3,1,1,1,3,3,1}{
    \node at (\value{gridcounterx}+.5,\value{gridcountery}+.5) {$\v$};
    \addtocounter{gridcounterx}{1}
    \ifthenelse{\value{gridcounterx}>7} 
    {
      \setcounter{gridcounterx}{3} 
      \addtocounter{gridcountery}{1}
    }{}
  }
\end{tikzpicture}
  }
  \caption{Cell to macrocell correspondence in the reduction from {\bf FSPP} to
  $\{1,2,3,4\}${\bf -FSPP}. Top: cells different from $v$. Bottom: $v$, with
  the new questioned cell highlighted.}
  \label{fig:1234}
\end{figure}

We now argue in details that $(c,v) \in$ {\bf FSPP} if and only if $(c',v')
\in$ $\{1,2,3,4\}${\bf -FSPP}, {\em i.e.} the reduction is correct. First,
$(c',v')$ is a valid instance of $\{1,2,3,4\}${\bf -FSPP} since $c'$ is a
finite $\{1,2,3,4\}$-simple configuration. Except for the questioned cell and
cells without grains (these latter having no influence on the dynamics),
there is a strict correspondence between the dynamics of $c$ and $c'$: if
vertex $(u_x,u_y)$ of $c$ fires at time $t$, then vertex $(5u_x+2,5u_y+2)$ of
$c'$ fires at time $5t$. Indeed, each background value 1 surrounding lines and
columns of 3 (which link the centers of macrocells) is neighbor of at most two
values 3 (even if we consider the macrocells neighbor to the macrocell
corresponding to $v$ when $c_v=0$), therefore none of them is ever fired and
the correspondence is strict. Regarding the new questioned cell $v'$ and its
associated macrocell, one can simply remark that for any value of $c_v$ the new
cell $v'$ is fired if and only if $4-c_v$ centers of neighboring macrocells are
fired, and that when this is not (yet) the case then no value $1$ in this
macrocell is fired. As a consequence we get the result.

\begin{proposition}\label{prop:1234}
  {\bf FSPP} $\leq^m_{\AC^0}$ $\{1,2,3,4\}${\bf -FSPP}.
\end{proposition}

\subsubsection{$\{0,2,3,4\}${\bf -FSPP}}

The reduction is defined as follows: given an instance $(c,v)$ of {\bf FSPP}, we
replace each vertex $(u_x,u_y) \in \Z^2$ of $c$ with a macrocell of size $5
\times 5$ whose bottom left corner is at position $(5u_x,5u_y)$. The cell to
macrocell correspondence is given on Figure~\ref{fig:0234}. This reduction can
be computed in constant parallel time, {\em i.e.} in $\AC^0$. Let us denote
$c'$ the obtained configuration with $v'$ the new questioned cell.

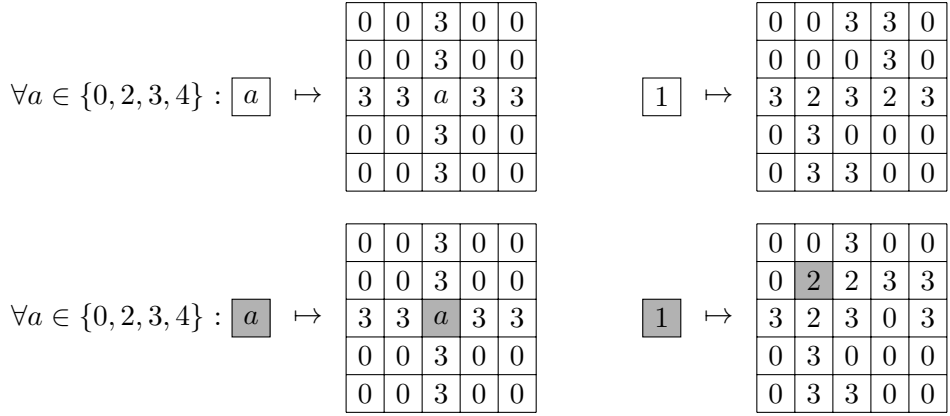
\begin{figure}
  \centerline{
    \begin{tikzpicture}[scale=.5]
  \node[left] at (0,.5) {$\forall a \in \{0,2,3,4\}:$};
  \draw (0,0) rectangle node{$a$} ++ (1,1);
  \node at (2,.5) {$\mapsto$};
  \draw (3,-2) grid (8,3);
  \setcounter{gridcounterx}{3} 
  \setcounter{gridcountery}{-2} 
  \foreach \v in {0,0,3,0,0,0,0,3,0,0,3,3,a,3,3,0,0,3,0,0,0,0,3,0,0}{
    \node at (\value{gridcounterx}+.5,\value{gridcountery}+.5) {$\v$};
    \addtocounter{gridcounterx}{1}
    \ifthenelse{\value{gridcounterx}>7} 
    {
      \setcounter{gridcounterx}{3} 
      \addtocounter{gridcountery}{1}
    }{}
  }
\end{tikzpicture}
    \hspace*{1cm}
    \begin{tikzpicture}[scale=.5]
  \draw (0,0) rectangle node{$1$} ++ (1,1);
  \node at (2,.5) {$\mapsto$};
  \draw (3,-2) grid (8,3);
  \setcounter{gridcounterx}{3} 
  \setcounter{gridcountery}{-2} 
  \foreach \v in {0,3,3,0,0,0,3,0,0,0,3,2,3,2,3,0,0,0,3,0,0,0,3,3,0}{
    \node at (\value{gridcounterx}+.5,\value{gridcountery}+.5) {$\v$};
    \addtocounter{gridcounterx}{1}
    \ifthenelse{\value{gridcounterx}>7} 
    {
      \setcounter{gridcounterx}{3} 
      \addtocounter{gridcountery}{1}
    }{}
  }
\end{tikzpicture}
  }
  \vspace*{1em}
  \centerline{
    \begin{tikzpicture}[scale=.5]
  \node[left] at (0,.5) {$\forall a \in \{0,2,3,4\}:$};
  \fill[black!30] (0,0) rectangle ++ (1,1); 
  \draw (0,0) rectangle node{$a$} ++ (1,1);
  \node at (2,.5) {$\mapsto$};
  \fill[black!30] (5,0) rectangle ++ (1,1); 
  \draw (3,-2) grid (8,3);
  \setcounter{gridcounterx}{3} 
  \setcounter{gridcountery}{-2} 
  \foreach \v in {0,0,3,0,0,0,0,3,0,0,3,3,a,3,3,0,0,3,0,0,0,0,3,0,0}{
    \node at (\value{gridcounterx}+.5,\value{gridcountery}+.5) {$\v$};
    \addtocounter{gridcounterx}{1}
    \ifthenelse{\value{gridcounterx}>7} 
    {
      \setcounter{gridcounterx}{3} 
      \addtocounter{gridcountery}{1}
    }{}
  }
\end{tikzpicture}
    \hspace*{1cm}
    \begin{tikzpicture}[scale=.5]
  \fill[black!30] (0,0) rectangle ++ (1,1); 
  \draw (0,0) rectangle node{$1$} ++ (1,1);
  \node at (2,.5) {$\mapsto$};
  \fill[black!30] (4,1) rectangle ++ (1,1); 
  \draw (3,-2) grid (8,3);
  \setcounter{gridcounterx}{3} 
  \setcounter{gridcountery}{-2} 
  \foreach \v in {0,3,3,0,0,0,3,0,0,0,3,2,3,0,3,0,2,2,3,3,0,0,3,0,0}{
    \node at (\value{gridcounterx}+.5,\value{gridcountery}+.5) {$\v$};
    \addtocounter{gridcounterx}{1}
    \ifthenelse{\value{gridcounterx}>7} 
    {
      \setcounter{gridcounterx}{3} 
      \addtocounter{gridcountery}{1}
    }{}
  }
\end{tikzpicture}
  }
  \caption{Cell to macrocell correspondence in the reduction from {\bf FSPP} to
  $\{0,2,3,4\}${\bf -FSPP}. Top: cells different from $v$. Bottom: $v$, with
  the new questioned cell highlighted.}
  \label{fig:0234}
\end{figure}

The argumentation regarding the correctness of this reduction is analogous to
the case of Proposition~\ref{prop:1234}, except that firings may not be
perfectly synchronized because of the macrocell corresponding to the value 1
doing some zigzag, but this has no consequence thanks to the so called {\em
abelian property} of sandpiles which still holds on freezing sandpiles (any
sequence of firings in $c$ is reproduced in $c'$, and conversely).

\begin{proposition}\label{prop:0234}
  {\bf FSPP} $\leq^m_{\AC^0}$ $\{0,2,3,4\}${\bf -FSPP}.
\end{proposition}

\subsubsection{$\{0,1,2,4\}${\bf -FSPP}}

The reduction is defined as follows: given an instance $(c,v)$ of {\bf FSPP}, we
replace each vertex $(u_x,u_y) \in \Z^2$ of $c$ with a macrocell of size $7
\times 7$ whose bottom left corner is at position $(7u_x,7u_y)$. The cell to
macrocell correspondence is given on Figure~\ref{fig:0124}. This reduction can
be computed in constant parallel time, {\em i.e.} in $\AC^0$. In the
constructed macrocells, each value 2 is neighbor of exactly one value 4, and
consequently all become value 3. The rest of the argumentation regarding the
correctness of this reduction is analogous to the case of
Proposition~\ref{prop:0234}.

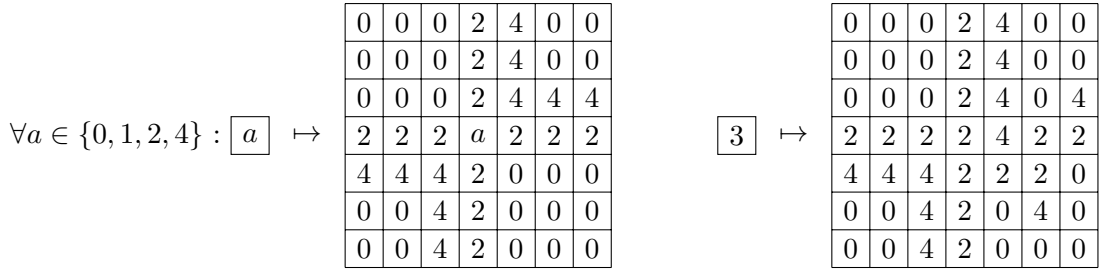
\begin{figure}
  \centerline{
    \begin{tikzpicture}[scale=.5]
  \node[left] at (0,.5) {$\forall a \in \{0,1,2,4\}:$};
  \draw (0,0) rectangle node{$a$} ++ (1,1);
  \node at (2,.5) {$\mapsto$};
  \draw (3,-3) grid (10,4);
  \setcounter{gridcounterx}{3} 
  \setcounter{gridcountery}{-3} 
  \foreach \v in {0,0,4,2,0,0,0,0,0,4,2,0,0,0,4,4,4,2,0,0,0,2,2,2,a,2,2,2,0,0,0,2,4,4,4,0,0,0,2,4,0,0,0,0,0,2,4,0,0}{
    \node at (\value{gridcounterx}+.5,\value{gridcountery}+.5) {$\v$};
    \addtocounter{gridcounterx}{1}
    \ifthenelse{\value{gridcounterx}>9} 
    {
      \setcounter{gridcounterx}{3} 
      \addtocounter{gridcountery}{1}
    }{}
  }
\end{tikzpicture}
    \hspace*{1cm}
    \begin{tikzpicture}[scale=.5]
  \draw (0,0) rectangle node{$3$} ++ (1,1);
  \node at (2,.5) {$\mapsto$};
  \draw (3,-3) grid (10,4);
  \setcounter{gridcounterx}{3} 
  \setcounter{gridcountery}{-3} 
  \foreach \v in {0,0,4,2,0,0,0,0,0,4,2,0,4,0,4,4,4,2,2,2,0,2,2,2,2,4,2,2,0,0,0,2,4,0,4,0,0,0,2,4,0,0,0,0,0,2,4,0,0}{
    \node at (\value{gridcounterx}+.5,\value{gridcountery}+.5) {$\v$};
    \addtocounter{gridcounterx}{1}
    \ifthenelse{\value{gridcounterx}>9} 
    {
      \setcounter{gridcounterx}{3} 
      \addtocounter{gridcountery}{1}
    }{}
  }
\end{tikzpicture}
  }
  \vspace*{1em}
  \caption{Cell to macrocell correspondence in the reduction from {\bf FSPP} to
  $\{0,1,2,4\}${\bf -FSPP}. Macrocells corresponding to the questioned cell $v$
  are identical, with the new questioned cell in the center (relative position
  $(3,3)$).}
  \label{fig:0124}
\end{figure}

\begin{proposition}\label{prop:0124}
  {\bf FSPP} $\leq^m_{\AC^0}$ $\{0,1,2,4\}${\bf -FSPP}.
\end{proposition}

\subsubsection{$\{1,2,4\}${\bf -FSPP}}

We give a reduction from $\{0,1,2,4\}${\bf -FSPP} to $\{1,2,4\}${\bf -FSPP}, by
replacing each vertex $(u_x,u_y) \in \Z^2$ of $c$ with a macrocell of size $5
\times 7$ whose bottom left corner is at position $(5u_x,7u_y)$. The cell to
macrocell correspondence is given on Figure~\ref{fig:124}. This reduction can
be computed in constant parallel time, {\em i.e.} in $\AC^0$.

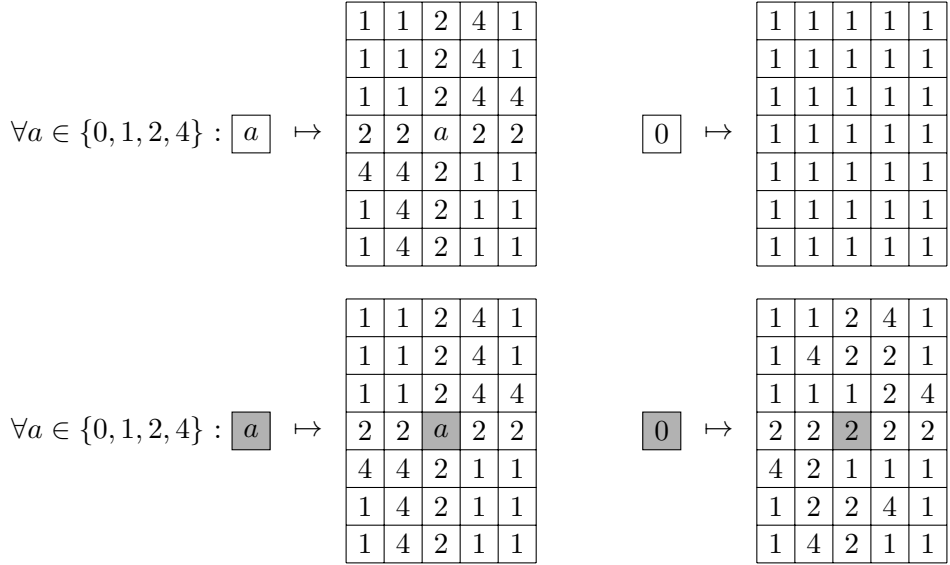
\begin{figure}
  \centerline{
    \begin{tikzpicture}[scale=.5]
  \node[left] at (0,.5) {$\forall a \in \{0,1,2,4\}:$};
  \draw (0,0) rectangle node{$a$} ++ (1,1);
  \node at (2,.5) {$\mapsto$};
  \draw (3,-3) grid (8,4);
  \setcounter{gridcounterx}{3} 
  \setcounter{gridcountery}{-3} 
  \foreach \v in {1,4,2,1,1,1,4,2,1,1,4,4,2,1,1,2,2,a,2,2,1,1,2,4,4,1,1,2,4,1,1,1,2,4,1}{
    \node at (\value{gridcounterx}+.5,\value{gridcountery}+.5) {$\v$};
    \addtocounter{gridcounterx}{1}
    \ifthenelse{\value{gridcounterx}>7} 
    {
      \setcounter{gridcounterx}{3} 
      \addtocounter{gridcountery}{1}
    }{}
  }
\end{tikzpicture}
    \hspace*{1cm}
    \begin{tikzpicture}[scale=.5]
  \draw (0,0) rectangle node{$0$} ++ (1,1);
  \node at (2,.5) {$\mapsto$};
  \draw (3,-3) grid (8,4);
  \setcounter{gridcounterx}{3} 
  \setcounter{gridcountery}{-3} 
  \foreach \v in {1,1,1,1,1,1,1,1,1,1,1,1,1,1,1,1,1,1,1,1,1,1,1,1,1,1,1,1,1,1,1,1,1,1,1}{
    \node at (\value{gridcounterx}+.5,\value{gridcountery}+.5) {$\v$};
    \addtocounter{gridcounterx}{1}
    \ifthenelse{\value{gridcounterx}>7} 
    {
      \setcounter{gridcounterx}{3} 
      \addtocounter{gridcountery}{1}
    }{}
  }
\end{tikzpicture}
  }
  \vspace*{1em}
  \centerline{
    \begin{tikzpicture}[scale=.5]
  \node[left] at (0,.5) {$\forall a \in \{0,1,2,4\}:$};
  \fill[black!30] (0,0) rectangle ++ (1,1); 
  \draw (0,0) rectangle node{$a$} ++ (1,1);
  \node at (2,.5) {$\mapsto$};
  \fill[black!30] (5,0) rectangle ++ (1,1); 
  \draw (3,-3) grid (8,4);
  \setcounter{gridcounterx}{3} 
  \setcounter{gridcountery}{-3} 
  \foreach \v in {1,4,2,1,1,1,4,2,1,1,4,4,2,1,1,2,2,a,2,2,1,1,2,4,4,1,1,2,4,1,1,1,2,4,1}{
    \node at (\value{gridcounterx}+.5,\value{gridcountery}+.5) {$\v$};
    \addtocounter{gridcounterx}{1}
    \ifthenelse{\value{gridcounterx}>7} 
    {
      \setcounter{gridcounterx}{3} 
      \addtocounter{gridcountery}{1}
    }{}
  }
\end{tikzpicture}
    \hspace*{1cm}
    \begin{tikzpicture}[scale=.5]
  \fill[black!30] (0,0) rectangle ++ (1,1); 
  \draw (0,0) rectangle node{$0$} ++ (1,1);
  \node at (2,.5) {$\mapsto$};
  \fill[black!30] (5,0) rectangle ++ (1,1); 
  \draw (3,-3) grid (8,4);
  \setcounter{gridcounterx}{3} 
  \setcounter{gridcountery}{-3} 
  \foreach \v in {1,4,2,1,1,1,2,2,4,1,4,2,1,1,1,2,2,2,2,2,1,1,1,2,4,1,4,2,2,1,1,1,2,4,1}{
    \node at (\value{gridcounterx}+.5,\value{gridcountery}+.5) {$\v$};
    \addtocounter{gridcounterx}{1}
    \ifthenelse{\value{gridcounterx}>7} 
    {
      \setcounter{gridcounterx}{3} 
      \addtocounter{gridcountery}{1}
    }{}
  }
\end{tikzpicture}
  }
  \caption{Cell to macrocell correspondence in the reduction from
  $\{0,1,2,4\}${\bf -FSPP} to $\{1,2,4\}${\bf -FSPP}. Top: cells different from
  $v$. Bottom: $v$, with the new questioned cell highlighted.}
  \label{fig:124}
\end{figure}

The argumentation regarding the correctness of this reduction is analogous to
the case of Proposition~\ref{prop:0124}, with the additional remark that some
values $1$ in the background may fire, without any side effect. Since $\AC^0$
is closed by composition, Proposition~\ref{prop:0124} gives the result.

\begin{proposition}\label{prop:124}
  {\bf FSPP} $\leq^m_{\AC^0}$ $\{1,2,4\}${\bf -FSPP}.
\end{proposition}

\subsubsection{$\{0,2,4\}${\bf -FSPP}}

We give a reduction from $\{0,2,3,4\}${\bf -FSPP} to $\{0,2,4\}${\bf -FSPP}, by
replacing each vertex $(u_x,u_y) \in \Z^2$ of $c$ with a macrocell of size $7
\times 7$ whose bottom left corner is at position $(7u_x,7u_y)$. The cell to
macrocell correspondence is the same as the one given on Figure~\ref{fig:0124}
from {\bf FSPP} to $\{0,1,2,4\}${\bf -FSPP}), except that the case $a=1$ is
removed (indeed, remark that macrocells do not make use of value 1). This
reduction can be computed in constant parallel time, {\em i.e.} in $\AC^0$.

The argumentation regarding the correctness of this reduction is analogous to
the case of Proposition~\ref{prop:0124}. Since $\AC^0$ is closed by
composition, Proposition~\ref{prop:0234} gives the result.

\begin{proposition}\label{prop:024}
  {\bf FSPP} $\leq^m_{\AC^0}$ $\{0,2,4\}${\bf -FSPP}.
\end{proposition}

\subsection{Perspectives on $\{1,3,4\}${\bf -FSPP} and $\{0,1,3,4\}${\bf -FSPP}}
\label{ss:0134}

Let us first notice that the complexity of predicting both models are
equivalent for $\AC^0$ reductions, with the cell to macrocell correspondence
given on Figure~\ref{fig:134-0134}.

\begin{figure}
  \centerline{
    \begin{tikzpicture}[scale=.5]
  \begin{scope}
    \draw (0,0) rectangle node{$0$} ++ (1,1);
    \node at (2,.5) {$\mapsto$};
    \draw (3,-1) grid (6,2);
    \setcounter{gridcounterx}{3} 
    \setcounter{gridcountery}{-1} 
    \foreach \v in {1,1,1,1,1,1,1,1,1}{
      \node at (\value{gridcounterx}+.5,\value{gridcountery}+.5) {$\v$};
      \addtocounter{gridcounterx}{1}
      \ifthenelse{\value{gridcounterx}>5} 
      {
        \setcounter{gridcounterx}{3} 
        \addtocounter{gridcountery}{1}
      }{}
    }
  \end{scope}
  \begin{scope}[shift={(0,-4)}]
    \draw (0,0) rectangle node{$1$} ++ (1,1);
    \node at (2,.5) {$\mapsto$};
    \draw (3,-1) grid (6,2);
    \setcounter{gridcounterx}{3} 
    \setcounter{gridcountery}{-1} 
    \foreach \v in {1,3,1,3,1,3,1,3,1}{
      \node at (\value{gridcounterx}+.5,\value{gridcountery}+.5) {$\v$};
      \addtocounter{gridcounterx}{1}
      \ifthenelse{\value{gridcounterx}>5} 
      {
        \setcounter{gridcounterx}{3} 
        \addtocounter{gridcountery}{1}
      }{}
    }
  \end{scope}
  \begin{scope}[shift={(7,0)}]
    \draw (0,0) rectangle node{$3$} ++ (1,1);
    \node at (2,.5) {$\mapsto$};
    \draw (3,-1) grid (6,2);
    \setcounter{gridcounterx}{3} 
    \setcounter{gridcountery}{-1} 
    \foreach \v in {1,3,1,3,3,3,1,3,1}{
      \node at (\value{gridcounterx}+.5,\value{gridcountery}+.5) {$\v$};
      \addtocounter{gridcounterx}{1}
      \ifthenelse{\value{gridcounterx}>5} 
      {
        \setcounter{gridcounterx}{3} 
        \addtocounter{gridcountery}{1}
      }{}
    }
  \end{scope}
  \begin{scope}[shift={(7,-4)}]
    \draw (0,0) rectangle node{$3$} ++ (1,1);
    \node at (2,.5) {$\mapsto$};
    \draw (3,-1) grid (6,2);
    \setcounter{gridcounterx}{3} 
    \setcounter{gridcountery}{-1} 
    \foreach \v in {1,3,1,3,4,3,1,3,1}{
      \node at (\value{gridcounterx}+.5,\value{gridcountery}+.5) {$\v$};
      \addtocounter{gridcounterx}{1}
      \ifthenelse{\value{gridcounterx}>5} 
      {
        \setcounter{gridcounterx}{3} 
        \addtocounter{gridcountery}{1}
      }{}
    }
  \end{scope}
\end{tikzpicture}
    \hspace*{1cm}
    \begin{tikzpicture}[scale=.5]
  \fill[black!30] (0,0) rectangle ++ (1,1); 
  \draw (0,0) rectangle node{$0$} ++ (1,1);
  \node at (2,.5) {$\mapsto$};
  \fill[black!30] (4,0) rectangle ++ (1,1); 
  \draw (3,-3) grid (10,4);
  \setcounter{gridcounterx}{3} 
  \setcounter{gridcountery}{-3} 
  \foreach \v in {1,1,1,3,1,1,1,1,1,1,3,1,1,1,1,1,1,3,1,1,1,3,1,3,1,3,3,3,1,4,1,3,1,1,1,1,1,1,3,1,1,1,1,1,1,3,1,1,1}{
    \node at (\value{gridcounterx}+.5,\value{gridcountery}+.5) {$\v$};
    \addtocounter{gridcounterx}{1}
    \ifthenelse{\value{gridcounterx}>9} 
    {
      \setcounter{gridcounterx}{3} 
      \addtocounter{gridcountery}{1}
    }{}
  }
\end{tikzpicture}
  }
  \caption{
    Cell to macrocell correspondence in the reduction from $\{0,1,3,4\}${\bf
    -FSPP} to $\{1,3,4\}${\bf -FSPP}. Left: correspondence when the questioned
    cell is not a $0$, in this case the new questioned cell is in the center of
    the corresponding macrocell. Right: if the questioned cell is a $0$ then we
    inflate all macrocells to be $7 \times 7$, and use the pictured macrocell
    to replace the questioned cell.
  }
  \label{fig:134-0134}
\end{figure}
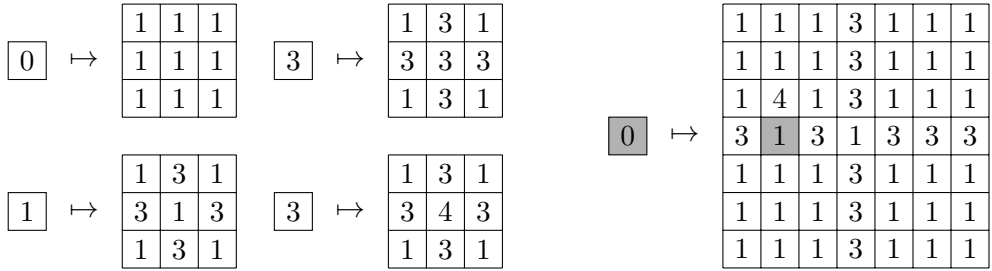

\begin{proposition}\label{prop:134-0134}
  $\{0,1,3,4\}${\bf -FSPP} $\leq^m_{\AC^0}$ $\{1,3,4\}${\bf -FSPP}.
\end{proposition}

When trying to find an $\NC$ algorithm to solve $\{1,3,4\}${\bf -FSPP}, our
attempts to adapt the reduction to strict majority employed for $\{0,1,4\}${\bf
-FSPP} in Subsection~\ref{ss:014} failed, because gadgets replacing a value $3$
seem to require degree five, though they can be made planar (see
Figure~\ref{fig:134-smaj}, but the case planar of degree at most five is
left open by~\cite{gmt13,gmmo17}).
Remark that we can answer efficiently in many cases using previous developments:
\begin{itemize}
  \item when the questioned cell is a $1$ in a cycle of $1$, or a $1$ on a path
    whose endpoints are connected to cycles of $1$ (decidable in $\NC^2$), then
    the answer is negative (as is the case for strict majority~\cite{gmt13});
  \item when the questioned cell is a $3$ connected via values $3$ to a $4$
    (decidable in $\NL$), then the answer is positive (as is the case for
    $\{0,3,4\}${\bf -FSPP} in Subsection~\ref{ss:034}).
\end{itemize}
We conjecture that the remaining cases are equivalent to planar {\em and-or}
freezing networks with fan in two and fan out one, but wires are undirected
(this is not a circuit) which leads to difficulties analogous to the general
case of $\{0,1,2,3,4\}${\bf -FSPP}, though interestingly in a seemingly more
restrictive setting.

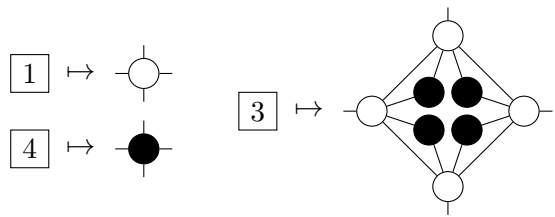
\begin{figure}
  \centerline{\begin{tikzpicture}[scale=.5]
  \tikzstyle{maj} = [draw,circle,inner sep=4pt]
  \draw (0,-.5) rectangle node{$1$} ++ (1,1);
  \node[left] at (2.5,0) {$\mapsto$};
  \node[maj] (n1) at (3.5,0) {};
  \draw (n1) -- ++ (.75,0);
  \draw (n1) -- ++ (-.75,0);
  \draw (n1) -- ++ (0,.75);
  \draw (n1) -- ++ (0,-.75);
  \draw (0,-2.5) rectangle node{$4$} ++ (1,1);
  \node[left] at (2.5,-2) {$\mapsto$};
  \node[maj,fill=black] (n4) at (3.5,-2) {};
  \draw (n4) -- ++ (.75,0);
  \draw (n4) -- ++ (-.75,0);
  \draw (n4) -- ++ (0,.75);
  \draw (n4) -- ++ (0,-.75);
  \draw (6,-1.5) rectangle node{$3$} ++ (1,1);
  \node[left] at (8.5,-1) {$\mapsto$};
  \begin{scope}[shift={(11.5,-1)}]
    \node[maj,fill=black] (in1) at (-.5,.5) {};
    \node[maj,fill=black] (in2) at (.5,.5) {};
    \node[maj,fill=black] (in3) at (.5,-.5) {};
    \node[maj,fill=black] (in4) at (-.5,-.5) {};
    \node[maj] (out1) at (0,2) {};
    \node[maj] (out2) at (2,0) {};
    \node[maj] (out3) at (0,-2) {};
    \node[maj] (out4) at (-2,0) {};
    \draw (out1) -- (in1);
    \draw (out1) -- (in2);
    \draw (out2) -- (in2);
    \draw (out2) -- (in3);
    \draw (out3) -- (in3);
    \draw (out3) -- (in4);
    \draw (out4) -- (in4);
    \draw (out4) -- (in1);
    \draw (out1) -- (out2);
    \draw (out2) -- (out3);
    \draw (out3) -- (out4);
    \draw (out4) -- (out1);
    \draw (out1) -- ++ (0,.75);
    \draw (out2) -- ++ (.75,0);
    \draw (out3) -- ++ (0,-.75);
    \draw (out4) -- ++ (-.75,0);
  \end{scope}
\end{tikzpicture}}
  \caption{
    Reduction from $\{1,3,4\}${\bf -FSPP} to strict majority dynamics on a
    planar undirected graph of degree at most 5, for which the complexity
    of prediction is open.
  }
  \label{fig:134-smaj}
\end{figure}

When trying to prove that {\bf FSPP} reduces to $\{0,1,3,4\}${\bf -FSPP}, we
failed to build a macrocell (with $0,1,3,4$) corresponding to a cell with $2$
sand grains (other elements are straightforward to design), though we found
some close constructions.
For example, the construction illustrated on
Figure~\ref{fig:0134-2ew} behaves almost as a value $2$, except that the
combination of {\em west} plus {\em east} signals does not trigger signals to the
{\em north} and {\em south} (any other combination of at least two signals
triggers signals to the remaining sides). We can
deduce that if the number of values $2$ is upper bounded by a polylogarithmic
function of the input's size, then there is an $\NC^1$ {\em truth-table}
reduction:
\begin{itemize}
  \item for each cell with value $2$ we try the macrocell of
    Figure~\ref{fig:0134-2ew} and the same rotated;
  \item the answer to {\bf FSPP} will be positive if and only if at least one
    combination (truth-table) of such macrocells for all values $2$ gives a
    positive answer;
  \item we need to compute the number $x$ of values $2$ (in $\NC^1$), and then
    $2^x$ transformations in parallel (a polynomial number, each in $\AC^0$).
\end{itemize}
We can also use the planar monotone circuit realizing
threshold function $T^{(4)}_2$ from~\cite[Figures~1 or~2]{c85} in order to create a
macrocell corresponding to a cell with $2$ sand grains (using diodes as on
Figure~\ref{fig:0134-2ew}), but the result given by the last gate is
``trapped'' inside the macrocell (signals are not sent to the remaining sides).
We can nevertheless deduce from such a construction that, if there is only
one value $2$ and if furthermore this is the questioned cell, then we have a
proper $\AC^0$ reduction.

\begin{figure}
  \centerline{\begin{tikzpicture}[scale=.5]
  \draw (-10,-10) grid (11,11);
  \setcounter{gridcounterx}{-10} 
  \setcounter{gridcountery}{-10} 
  \foreach \v in
  { , , , , , , , , , ,3, , , , , , , , , , ,
    , , , , , , ,3,3, ,3, ,3,3, , , , , , , ,
    , , , , ,3,3,1,3,3,3,3,1,3,3, , , , , , ,
    , , , , ,3, ,4, , , , ,4, ,3, , , , , , ,
    , , , , ,3, , , , , , , , ,3, , , , , , ,
    , ,3,3,3,3,3,3,3,3, , , , ,3, , , , , , ,
    , ,3, , , , , , ,3, ,3,3,3,3,3,3,3,3, , ,
    ,3,1,4, , , , , ,3, ,3, , , , , , ,3,3, ,
    ,3,3, , , , , , ,3, ,3, , , , , ,4,1,3, ,
    , ,3, , , , , , ,3,3,1,4, , , , , ,3, , ,
   3,3,3, , , , , , ,3, ,3, , , , , , ,3,3,3,
    , ,3, , , , , ,4,1,3,3, , , , , , ,3, , ,
    ,3,1,4, , , , , ,3, ,3, , , , , , ,3,3, ,
    ,3,3, , , , , , ,3, ,3, , , , , ,4,1,3, ,
    , ,3,3,3,3,3,3,3,3, ,3, , , , , , ,3, , ,
    , , , , , ,3, , , , ,3,3,3,3,3,3,3,3, , ,
    , , , , , ,3, , , , , , , , ,3, , , , , ,
    , , , , , ,3, ,4, , , , ,4, ,3, , , , , ,
    , , , , , ,3,3,1,3,3,3,3,1,3,3, , , , , ,
    , , , , , , ,3,3, ,3, ,3,3, , , , , , , ,
    , , , , , , , , , ,3, , , , , , , , , , ,
  }{
    \node at (\value{gridcounterx}+.5,\value{gridcountery}+.5) {$\v$};
    \addtocounter{gridcounterx}{1}
    \ifthenelse{\value{gridcounterx}>10} 
    {
      \setcounter{gridcounterx}{-10} 
      \addtocounter{gridcountery}{1}
    }{}
  }
  \draw[very thick] (-3,10) -- ++ (2,0) -- ++ (0,-1) -- ++ (1,0) -- ++ (0,-1) -- ++ (-1,0) -- ++ (0,-1) -- ++ (-1,0) -- ++ (0,1) -- ++ (-2,0) -- ++ (0,1) -- ++ (1,0) -- cycle;
  \draw[very thick,-latex] (-3,10.5) -- ++ (2,0);
  \draw[very thick] (2,10) -- ++ (2,0) -- ++ (0,-1) -- ++ (1,0) -- ++ (0,-1) -- ++ (-1,0) -- ++ (0,-1) -- ++ (-1,0) -- ++ (0,1) -- ++ (-2,0) -- ++ (0,1) -- ++ (1,0) -- cycle;
  \draw[very thick,-latex] (2,10.5) -- ++ (2,0);
  \draw[very thick] (8,5) -- ++ (1,0) -- ++ (0,-1) -- ++ (1,0) -- ++ (0,-2) -- ++ (-1,0) -- ++ (0,-1) -- ++ (-1,0) -- ++ (0,2) -- ++ (-1,0) -- ++ (0,1) -- ++ (1,0) -- cycle;
  \draw[very thick,-latex] (10.5,2) -- ++ (0,2);
  \draw[very thick] (8,0) -- ++ (1,0) -- ++ (0,-1) -- ++ (1,0) -- ++ (0,-2) -- ++ (-1,0) -- ++ (0,-1) -- ++ (-1,0) -- ++ (0,2) -- ++ (-1,0) -- ++ (0,1) -- ++ (1,0) -- cycle;
  \draw[very thick,-latex] (10.5,-3) -- ++ (0,2);
  \draw[very thick] (2,-6) -- ++ (1,0) -- ++ (0,-1) -- ++ (2,0) -- ++ (0,-1) -- ++ (-1,0) -- ++ (0,-1) -- ++ (-2,0) -- ++ (0,1) -- ++ (-1,0) -- ++ (0,1) -- ++ (1,0) -- cycle;
  \draw[very thick,-latex] (4,-9.5) -- ++ (-2,0);
  \draw[very thick] (-3,-6) -- ++ (1,0) -- ++ (0,-1) -- ++ (2,0) -- ++ (0,-1) -- ++ (-1,0) -- ++ (0,-1) -- ++ (-2,0) -- ++ (0,1) -- ++ (-1,0) -- ++ (0,1) -- ++ (1,0) -- cycle;
  \draw[very thick,-latex] (-1,-9.5) -- ++ (-2,0);
  \draw[very thick] (-8,0) -- ++ (1,0) -- ++ (0,-2) -- ++ (1,0) -- ++ (0,-1) -- ++ (-1,0) -- ++ (0,-1) -- ++ (-1,0) -- ++ (0,1) -- ++ (-1,0) -- ++ (0,2) -- ++ (1,0) -- cycle;
  \draw[very thick,-latex] (-9.5,-1) -- ++ (0,-2);
  \draw[very thick] (-8,5) -- ++ (1,0) -- ++ (0,-2) -- ++ (1,0) -- ++ (0,-1) -- ++ (-1,0) -- ++ (0,-1) -- ++ (-1,0) -- ++ (0,1) -- ++ (-1,0) -- ++ (0,2) -- ++ (1,0) -- cycle;
  \draw[very thick,-latex] (-9.5,4) -- ++ (0,-2);
\end{tikzpicture}}
  \caption{
    White cells have no sand grain ($0$), and diode mechanisms are highlighted.
    Macrocell with $0,1,3,4$ corresponding to a value $2$, except that
    the {\em west} plus {\em east} signals do not trigger signals to
    the {\em north} and {\em south} sides.
    Any other combination of at least two signals triggers signals to
    the remaining sides.
    The same macrocell rotated by 90 degrees is only missing the {\em north}
    plus {\em south} combination.
  }
  \label{fig:0134-2ew}
\end{figure}
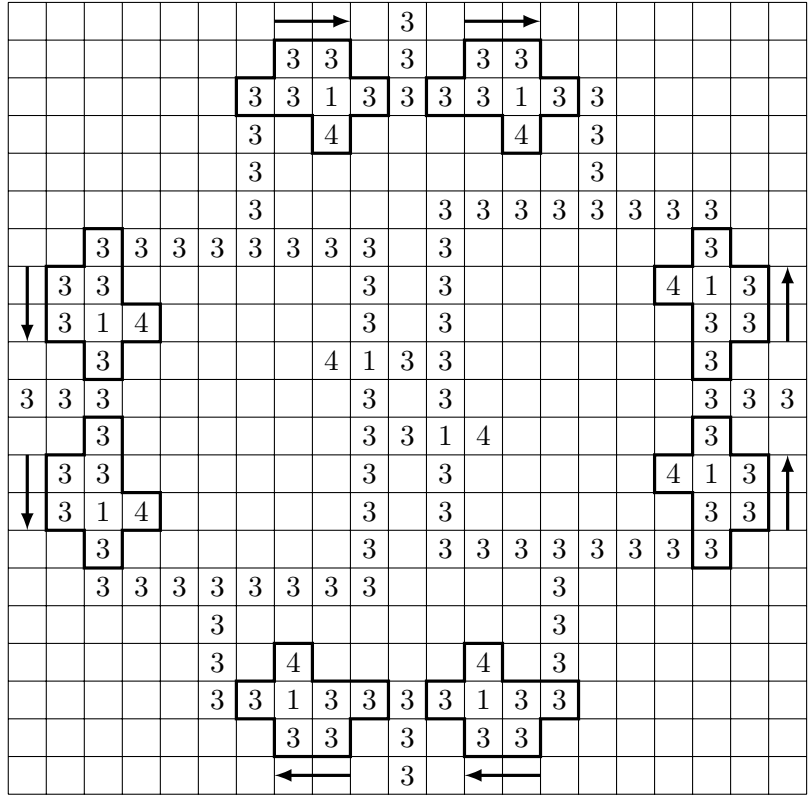

\section{Conclusion}

The freezing world allows to make insightful progresses related to difficult questions.
Exploiting the formal connections with threshold Boolean functions established
by Proposition~\ref{prop:bool}, Theorems~\ref{theorem:nc}
and~\ref{theorem:fspp} characterise the computational complexity of all but two
restrictions of {\bf freezing sandpile prediction problem} ({\bf FSPP}):
either the problem is as hard as unrestricted {\bf FSPP};
or it is proven to be in $\NC$ (or below).
The results are displayed on Table~\ref{table:results}.

\begin{table}
  \centerline{
    \def\arraystretch{1.3}
    \begin{tabular}{|c|c|c||c||c|}
      \hline
      $\AC^0$ & $\NL$ & $\NC^2$ & {\bf FSPP}$\leq^m_{\AC^0}$ & Open\\
      \hline
      $\{0,4\}$ & $\{0,3,4\}$ & $\{1,4\}$ & $\{1,2,3,4\}$ & $\{0,1,3,4\}$\\[-4pt]
      & & $\{0,1,4\}$ & $\{0,2,3,4\}$ & $\{1,3,4\}$\\[-4pt]
      & & $\{2,4\}$ & $\{0,1,2,4\}$ & \\[-4pt]
      & & $\{2,3,4\}$ & $\{1,2,4\}$& \\[-4pt]
      & & & $\{0,2,4\}$ & \\
      \hline
    \end{tabular}
  }
  \caption{Summary of Theorems~\ref{theorem:nc},~\ref{theorem:fspp} and
  Open question~\ref{question:0134}.}
  \label{table:results}
\end{table}

The results show interesting fine-grained view on necessary and sufficient
conjunctions of elements (values among $\{0,1,2,3,4\}$) for the dynamics to be
``as expressive as'' {\bf FSPP}. We propose three remarks.

First, in~\cite{gmt13} it is proven that the prediction problem on the non-strict
majority cellular automata is $\Poly$-hard for the family of graphs with
maximum degree at least $4$, and that it is in $\NC$ for the family of
graphs with maximum degree at most $3$. In~\cite{gmmo17} it is proven that
the same problem is in $\NC$ when the graph restricted to the two
dimensional grid with von Neumann neighborhood (a particular case of
regular graph where each vertex has degree $4$), which corresponds to
$\{2,4\}${\bf -FSPP}. According to Section~\ref{s:bool}, the problem
$\{0,2,4\}${\bf -FSPP} introduces a new refinement: when restricted to the
two dimensional grid with von Neumann neighborhood on which some vertices
are somehow {\em removed} (with sand value 0), the problem becomes as hard
as {\bf FSPP}.

Second, in the reduction of Proposition~\ref{prop:0124} (and some subsequent ones) it
seems important to have many values $4$. What is the computational complexity
of the weak prediction problem (given a finite {\em stable} configuration,
plus only one sand grain addition, namely {\bf $1^\text{st}$-col-S-PRED}
of the survey~\cite{fp19}) in this case?

Third, the Open question~\ref{question:0134} puts in light a surprisingly complex
refinement, where forbidding only the value $2$ seems to decrease the
expressiveness of the model, yet not flattening it to another known case. Could
it be that, if $\NC \neq \Poly$, then {\bf FSPP} and $\{0,1,3,4\}${\bf -FSPP}
would belong to different intermediate classes strictly between $\NC$ and
$\Poly$ (which would exist according to an analog of Ladner's theorem~\cite{v90})?

The present work circumvents the question of whether
{\bf FSPP} itself is in $\NC$ or $\Poly$-hard. One can implement conjunctions,
disjunctions, but the relationship between the impossibility of crossing
wires~\cite{gg06} and the possibility of using undirected wires, or even other
forms of signal implementation, leaves open its reduction to MPCVP~\cite{y91}
(prediction in $\NC$), or the possibility to implement non-planar or
non-monotone gates~\cite{l75} ($\Poly$-hard prediction). The general case of
{\bf FSPP} reduces to $\{0,2,4\}${\bf -FSPP}, could that help in order to find
an efficient algorithm?
Advances on {\bf FSPP} would constitute great insights for the classical
sandpile prediction problem ({\bf SPP}) in two dimensions, left open in the
original paper by Moore and Nilsson~\cite{mn99}, even though some relationship
between {\bf FSPP} and {\bf SPP} is still to be formally established.

Finally, the relationship between threshold functions and cell's sand content
opens perspectives on the prediction of Boolean functions on the grid: in the
freezing and non-freezing worlds, what are the necessary and sufficient elements
in order to have easy/hard prediction problems?

\section*{Acknowledgments}

This research was partially supported by
ANID via PAI + Convocatoria Nacional Subvenci\'on a la Incorporaci\'on en la Academia A\~no 2017 + PAI77170068 (P.M.),
FONDECYT 11190482 (P.M.),
FONDECYT 1200006 (E.G., P.M.),
STIC- AmSud CoDANet project 88881.197456/2018-01 (E.G., P.M., K.P.),
ANR-18-CE40-0002 FANs (K.P.).

\bibliographystyle{plain}
\bibliography{biblio}

\begin{thebibliography}{10}

\bibitem{btw87}
P.~Bak, C.~Tang, and K.~Wiesenfeld.
\newblock Self-organized criticality: An explanation of the 1/\textit{f} noise.
\newblock {\em Physical Review Letter}, 59:381--384, 1987.

\bibitem{b71}
E.~R. Banks.
\newblock {\em {Information processing and transmission in cellular automata}}.
\newblock PhD thesis, Massachusetts Institute of Technology, 1971.

\bibitem{bmot18}
F.~Becker, D.~Maldonado, N.~Ollinger, and G.~Theyssier.
\newblock {Universality in Freezing Cellular Automata}.
\newblock In {\em Proceedings of CiE'2018}, volume 10936 of {\em LNCS}, pages
  50--59, 2018.

\bibitem{fp19}
E.~Formenti and K.~Perrot.
\newblock {How Hard is it to Predict Sandpiles on Lattices? A Survey}.
\newblock {\em Fondamenta Informaticae}, 171:189--219, 2019.

\bibitem{gg06}
A.~Gajardo and E.~Goles.
\newblock Crossing information in two-di\-mensional sandpiles.
\newblock {\em Theoretical Computer Science}, 369(1-3):463--469, 2006.

\bibitem{glmmp04}
E.~Goles, M.~Latapy, C.~Magnien, M.~Morvan, and H.~D. Phan.
\newblock Sandpile models and lattices: a comprehensive survey.
\newblock {\em Theoretical Computer Science}, 322(2):383--407, 2004.

\bibitem{gmmo17}
E.~Goles, D.~Maldonado, P.~Montealegre, and N.~Ollinger.
\newblock On the computational complexity of the freezing non-strict majority
  automata.
\newblock In {\em Proceedings of AUTOMATA'2017}, volume 10248 of {\em LNCS},
  pages 109--119, 2017.

\bibitem{gm97}
E.~Goles and M.~Margenstern.
\newblock Universality of the chip-firing game.
\newblock {\em Theoretical Computer Science}, 172(1-2):121--134, 1997.

\bibitem{gm90}
E.~Goles and S.~Mart{\'{\i}}nez.
\newblock {\em Neural and Automata Networks}.
\newblock Springer Netherlands, 1990.

\bibitem{gm14}
E.~Goles and P.~Montealegre.
\newblock Computational complexity of threshold automata networks under
  different updating schemes.
\newblock {\em Theoretical Computer Science}, 559:3--19, 2014.

\bibitem{gm15}
E.~Goles and P.~Montealegre.
\newblock {The complexity of the majority rule on planar graphs}.
\newblock {\em {Advances in Applied Mathematics}}, 64:111--123, 2015.

\bibitem{gm16}
E.~Goles and P.~Montealegre.
\newblock A fast parallel algorithm for the robust prediction of the
  two-dimensional strict majority automaton.
\newblock In {\em Proceedings of ACRI'2016}, pages 166--175, 2016.

\bibitem{gmpt17}
E.~Goles, P.~Montealegre, K.~Perrot, and G.~Theyssier.
\newblock {On the complexity of two-dimensional signed majority cellular
  automata}.
\newblock {\em Journal of Computer and System Sciences}, 91:1--32, 2017.

\bibitem{gmt13}
E.~Goles, P.~Montealegre-Barba, and I.~Todinca.
\newblock The complexity of the bootstraping percolation and other problems.
\newblock {\em Theoretical Computer Science}, 504:73--82, 2013.

\bibitem{got15}
E.~Goles, N.~Ollinger, and G.~Theyssier.
\newblock {Introducing Freezing Cellular Automata}.
\newblock In {\em Proceedings of AUTOMATA'15}, volume~24 of {\em TUCS Lecture
  Notes}, pages 65--73, 2015.

\bibitem{ghr95}
R.~Greenlaw, H.~J. Hoover, and W.~L. Ruzzo.
\newblock {\em Limits to Parallel Computation: P-Completeness Theory}.
\newblock Oxford University Press, Inc., 1995.

\bibitem{l75}
R.~E. Ladner.
\newblock The circuit value problem is log space complete for {P}.
\newblock {\em SIGACT News}, 7(1):18--20, 1975.

\bibitem{m93}
J.~Matcha.
\newblock The computational complexity of pattern formation.
\newblock {\em Journal of Statistical Physics}, 70(3):949--966, 1993.

\bibitem{c85}
W.~F. McColl.
\newblock On the planar monotone computation of threshold functions.
\newblock In {\em Proceedings of STACS'85}, volume 182 of {\em LNCS}, pages
  219--230, 1985.

\bibitem{m97}
C.~Moore.
\newblock Majority-vote cellular automata, ising dynamics, and p-completeness.
\newblock {\em Journal of Statistical Physics}, 88(3):795--805, 1997.

\bibitem{mn99}
C.~Moore and M.~Nilsson.
\newblock The computational complexity of sandpiles.
\newblock {\em Journal of Statistical Physics}, 96:205--224, 1999.

\bibitem{mn97}
C.~Moore and M.~G Nordahl.
\newblock Predicting lattice gases is p-complete.
\newblock Technical report, Santa Fe Institute Working Paper 97-04-034, 1997.

\bibitem{nw06}
T.~Neary and D.~Woods.
\newblock P-completeness of cellular automaton rule 110.
\newblock In {\em Proceedings of ICALP'2006}, volume 4051 of {\em LNCS}, pages
  132--143, 2006.

\bibitem{np18}
V.-H. Nguyen and K.~Perrot.
\newblock {Any shape can ultimately cross information on two-dimensional
  abelian sandpile models}.
\newblock In {\em Proceedings of AUTOMATA'2018}, volume 10875 of {\em LNCS},
  pages 127--142, 2018.

\bibitem{ot19}
N.~Ollinger and G.~Theyssier.
\newblock Freezing, bounded-change and convergent cellular automata.
\newblock {\em Preprint on arXiv:1908.06751}, 2019.

\bibitem{v90}
H.~Vollmer.
\newblock The gap-language-technique revisited.
\newblock In {\em Proceedings of CSL'90}, volume 533 of {\em LNCS}, pages
  389--399, 1990.

\bibitem{y91}
H.~Yang.
\newblock An {NC} algorithm for the general planar monotone circuit value
  problem.
\newblock In {\em Proceedings of IPDPS'91}, pages 196--203, 1991.

\end{thebibliography}

\end{document}